\documentclass[english]{DDCLSconf}
\usepackage[comma,numbers,square,sort&compress]{natbib}
\usepackage{epstopdf}
\usepackage{amsmath, amsfonts, amssymb, amsthm}
\usepackage{graphicx, epstopdf}
\usepackage{url, verbatim, xcolor}
\usepackage{booktabs, array, stfloats}
\usepackage{soul}
\usepackage{xcolor}
\sethlcolor{yellow}
\usepackage{graphicx}  % 用于插入图片
\usepackage{subfig}    % 关键：用于支持 \subfloat 命令

\newtheorem{theorem}{Theorem}

\newtheorem{assumption}{Assumption}
\theoremstyle{remark}
\newtheorem{remark}{Remark}
\begin{document}

\title{On Switched Event-triggered Full State-constrained Formation Control for Multi-vehicle Systems}

\author{Zihan Li\aref{1},
        Ziming Wang\aref{2,3},
        Xin Wang\aref{1}\textsuperscript{*}}

\affiliation[1]{Southwest University, Chongqing, 400715, China.
        \email{xinwangswu@163.com}}
\affiliation[2]{The Hong Kong University of Science and Technology (Guangzhou), Guangzhou, 511453, China.}
\affiliation[3]{Tsinghua University, Beijing, 100084, China.}
\maketitle

\begin{abstract}
Vehicular formation control is an important component of intelligent transportation systems (ITSs). In practical implementations, the controller design needs to satisfy multiple state constraints, including inter-vehicle spacing and vehicle speed. When system states approach the constraint boundaries, control singularity and excessive control effort may arise, which limits the practical applicability of existing methods. To address this problem, this paper investigates a class of nonlinear vehicular formation systems for autonomous vehicles (AVs) with uncertain dynamics and develops a switched event-triggered control framework. A smooth nonlinear mapping is first introduced to transform the constrained state space into an unconstrained one, thereby avoiding singularity near the constraint boundaries. A radial basis function neural network (RBFNN) is then employed to approximate the unknown nonlinear dynamics online, based on which an adaptive controller is constructed via the backstepping technique. In addition, a switched event-triggered mechanism (SETM) is designed to increase the control update frequency during the transient stage and reduce the communication burden during the steady-state stage. Lyapunov-based analysis proves that all signals in the closed-loop system remain uniformly bounded and that Zeno behavior is excluded. Simulation results verify that the proposed method achieves stable platoon formation under prescribed state constraints while significantly reducing communication updates.
\end{abstract}

\keywords{Formation control, State constrained system, Event-triggered control (ETC), Autonomous vehicles (AVs)}

\section{Introduction}
For decades, formation control in multi-agent systems (MASs) has attracted widespread attention due to its applications in autonomous vehicles (AVs), collaborative robots, and intelligent transport systems \cite{liu2026predefined,ref_fei_2024}. In these scenarios, vehicle platooning is regarded as a key technology capable of improving traffic efficiency and reducing energy consumption \cite{wang2024state}, while also enhancing road safety \cite{ref_liu_2025,liao2026event}. During convoy operation, each vehicle must not only track the target position but also maintain a safe following distance and comply with speed limits. If these constraints are violated, collisions or convoy instability may occur. Therefore, strict adherence to state constraints is essential for the practical application of autonomous convoy systems \cite{ref_wang_emerg_2024,ref_yang_2025}.

However, the barrier Lyapunov function (BLF) method \cite{ref_sachan_2018,ref_liu_automatica_2018,ref_wang_ccdc_2024} often faces structural singularity issues near constraint boundaries. When states approach these boundaries, penalty terms can increase sharply, leading to excessive control inputs or numerical sensitivity \cite{ref_dong_2020}. In dynamic scenarios such as fleet reorganization, this may cause actuator saturation and degrade performance \cite{ref_wang_tcas_2024}. Consequently, ensuring safety constraints while avoiding control singularities remains a significant research challenge \cite{ref_xing_2017,ref_wang_smc_2025,ref_gong_2025,ref_min_2020,ref_wang_cybern_2025,li2026neural}.

Constraint transformation methods have recently received increasing attention \cite{ref_liang_2026,liu2026approximation}. By mapping constrained physical states into an unconstrained space, they can avoid the singularity of conventional barrier functions \cite{ref_qiao_2022}. Control design in the transformed space also reduces numerical difficulty and simplifies stability analysis \cite{ref_wang_case_2025,ref_wang_arxiv_2025}. Motivated by this idea, this paper uses a smooth nonlinear mapping to handle spacing and velocity constraints \cite{ref_zhang_2023}.

In addition to safety constraints, communication efficiency is also a key issue in distributed fleet control \cite{ref_wen_2018}. Conventional methods typically rely on periodic sampling and continuous communication, which can easily lead to the transmission of large amounts of redundant data \cite{ref_wu_2022}. Event-triggered control (ETC) updates control signals only when trigger conditions are met, thereby significantly reducing the communication burden \cite{ref_nguyen_2020,jiang2026when,ref_xue_2024,ref_wang_amc_2021,wang2025fixed}.

Most ETC schemes use a fixed triggering threshold \cite{ref_wen_2018,ref_nguyen_2020,ref_wang_amc_2021}. For nonlinear vehicular platoons with full-state constraints, this may cause a trade-off between transient performance and communication efficiency: a small threshold increases updates, whereas a large one slows the response. The problem becomes more pronounced when state constraints, model uncertainty, and stage-dependent errors are considered together \cite{ref_xue_2024,ref_xu_2026}. To address this issue, this paper introduces a switched event-triggered mechanism that adjusts the threshold according to the current error magnitude.

Existing studies mainly emphasize either communication efficiency or constrained control \cite{wang2024state,ref_liu_automatica_2018,ref_wang_ccdc_2024,ref_wang_tcas_2024}. For vehicle formations with full-state constraints and uncertain dynamics, it is still difficult to develop a unified framework that balances constraint satisfaction, tracking performance, robustness, and communication efficiency \cite{ref_wang_smc_2025,liu2026approximation}.

To address this issue, this paper proposes a robust adaptive ETC scheme for nonlinear vehicle formations with full-state constraints. A smooth nonlinear mapping transforms constrained states into unconstrained coordinates, avoiding control singularities. An RBFNN approximates unknown dynamics, and a switched event-triggered mechanism (SETM) adjusts the update frequency. The main contributions are as follows:

\begin{itemize}
    \item We develop a smooth diffeomorphism-based mapping that transforms constrained states into an unconstrained domain, ensuring velocity and spacing constraints without control singularities. 

    \item We develop an adaptive controller combining mapping gains and RBFNNs to compensate for uncertainties, and prove that all closed-loop signals are uniformly ultimately bounded (UUB). 

    \item We design a SETM that adjusts update thresholds to reduce communication load and actuator wear while excluding Zeno behavior. 
\end{itemize}

% \begin{figure}
%     \centering
%     \includegraphics[width=1\linewidth]{Square Formation.pdf}
%     \caption{AVs formation control in square formation.}
%     \label{fig:placeholder}
% \end{figure}

Section \ref{sec:preliminary} introduces graph theory, vehicle models, state constraints and control objectives. Section \ref{sec:method} details an adaptive control framework that includes the full-state constraint mapping, the neural network, and the SETM. Section \ref{sec:stability} proves the system’s stability using Lyapunov theory. Section \ref{sec:example} verifies the algorithm’s effectiveness through a simulation. Finally, Section \ref{sec:conclusion} summarizes the paper.

\section{Preliminaries and Problem Formulation}~\label{sec:preliminary}

\subsection{Graph theory and square formation} \label{subsec:Graph_Formation}
Consider a MAS consisting of $N$ AVs. The communication topology is described by a connected undirected graph $\mathcal{G}=(\mathcal{V},\mathcal{E},\mathcal{A})$, where $\mathcal{V}=\{1,\dots,N\}$ is the node set, $\mathcal{E}$ is the edge set, and $\mathcal{A}=[a_{ij}]$ is the adjacency matrix. If vehicles $i$ and $j$ exchange information, then $a_{ij}=a_{ji}>0$; otherwise $a_{ij}=0$. The neighbor set of vehicle $i$ is defined as $\mathcal{N}_i=\{j\in\mathcal{V}\mid a_{ij}>0\}$. The Laplacian matrix is given by $L=D-\mathcal{A}$, where $D=\mathrm{diag}(d_1,\dots,d_N)$ with $d_i=\sum_{j=1}^N a_{ij}$. 

Let $p_i=[x_i,y_i]^T$ and $v_i=[v_{xi},v_{yi}]^T$ denote the position and velocity of vehicle $i$, respectively. Each vehicle is assigned a desired relative coordinate $p_i^d=[x_i^d,y_i^d]^T$. For a four-vehicle square formation, the desired positions are chosen as
$p_1^d=(0,0)$, $p_2^d=(d,0)$, $p_3^d=(d,d)$, and $p_4^d=(0,d)$. The formation tracking error based on neighbor information is defined as
\begin{equation}
z_{1i}=\sum_{j\in\mathcal{N}_i}a_{ij}\left[(p_i-p_j)-(p_i^d-p_j^d)\right].
\label{eq:neighbor_error}
\end{equation}
where $z_{1i}=0$ for all $i\in\mathcal{V}$, and $p_i-p_j=p_i^d-p_j^d$, indicating that the desired square formation is achieved.

\subsection{Vehicle dynamics}
\label{subsec:Vehicle_Dynamics}
Consider the second-order nonlinear dynamics of the $i$th AV. Let $p_i=[x_i,y_i]^T$ and $v_i=[v_{xi},v_{yi}]^T$ denote the position and velocity, respectively.
Define the state vector $x_i = [p_i, v_i]^T$. The vehicle dynamics are described by the following equations:
\begin{equation}
\begin{cases}
\dot{p}_i = v_i \\
\dot{v}_i = f_i(x_i) + g_i(x_i)u_i + d_i(t)
\end{cases}
\label{eq:vehicle_dynamics}
\end{equation}
where $u_i$ is the control input, $f_i(x_i)$ and $g_i(x_i)$ are unknown smooth nonlinear functions, and $d_i(t)$ represents a bounded external disturbance. To approximate the unknown nonlinear terms $f_i(x_i)$, a RBFNN is employed
\begin{equation}
f_i(x_i)=W_i^T\phi_i(x_i)+\varepsilon_i
\label{eq:RBF_definition}
\end{equation}
where $W_i$ is the ideal weight matrix, $\phi_i(x_i)$ is a Gaussian vector, and $\varepsilon_i$ represents the approximation error. Suppose the error upper bound as $\|\varepsilon_i\|\le\bar{\varepsilon}_i$, where $\bar{\varepsilon}_i$ is a constant.

\subsection{State constraints}
\label{subsec:State_Constraints}
Each vehicle is subject to spacing and speed constraints.
Let $d_i = p_{i-1} - p_i$ denote the longitudinal spacing, and $v_i$ denote the speed. The safety constraints are $d_{\min} < d_i < d_{\max}$, $v_{\min} < v_i < v_{\max}$. To incorporate the speed constraints into the controller design, we introduce a smooth diffeomorphism that maps the constrained speed $v_i$ to the unconstrained variable $s_{v,i}$, then we have
\begin{equation}
s_{v,i}=\ln\left(\frac{v_i-v_{\min}}{v_{\max}-v_i}\right)
\label{eq:sv_mapping}
\end{equation}

The logarithmic mapping is valid for $v_{\min} < v_i < v_{\max}$. Under feasible initial conditions, the inverse mapping and Lyapunov analysis keep the closed-loop trajectories inside the admissible set. In practice, a projection or saturation operator can be used before evaluating the mapping to avoid numerical failure, such that
\begin{equation}
    v_i \in [v_{\min} + \eta,\, v_{\max} - \eta]
\end{equation}
where $\eta > 0$ is a small margin.

Based on the formation tracking error $z_{1i}$ defined in Equation \ref{eq:neighbor_error}, the second-order error is constructed as follows:
\begin{equation}
z_{2i}=s_{v,i}-\alpha_i
\label{eq:z2_constrained}
\end{equation}
where $\alpha_i$ represents the virtual control input.

\subsection{Control Objectives}

The objective is to design a controller $u_i$ and an associated SETM for each vehicle so that the MAS achieves the desired square formation under all state constraints. A nonlinear mapping enforces distance and velocity constraints, an RBFNN compensates for unknown dynamics and disturbances, and the SETM reduces communication load while preserving control performance. Lyapunov analysis guarantees UUB of all closed-loop signals and excludes Zeno behavior. 

\begin{assumption}
The fleet communication topology $\mathcal{G}$ is an undirected connected network, i.e., for any vehicles $i$ and $j$, there exists at least one communication path between them.
\label{assum:graph}
\end{assumption}

\begin{assumption}
Both the external disturbance $d_i(t)$ and the RBFNN approximation error $\varepsilon_i(x_i)$ are bounded. Specifically, there exist constants $\bar{d}_i$ and $\bar{\varepsilon}_i$ such that
\begin{equation}
\|d_i(t)\| \le \bar{d}_i, \quad \|\varepsilon_i(x_i)\| \le \bar{\varepsilon}_i, \quad \forall t \ge 0.
\end{equation}
\label{assum:disturbance}
\end{assumption}

\begin{assumption}
The control input gain matrix $g_i(x_i)$ in Equation (1) is nonsingular and invertible for all $x_i \in \mathbb{R}^n$. Furthermore, there exist constants $g_{\min}$ and $g_{\max}$ such that
\begin{equation}
0 < g_{\min} \le \|g_i(x_i)\| \le g_{\max}.
\end{equation}
\label{assum:gain}
\end{assumption}

\section{Main Results}\label{sec:method}

\subsection{Transformed system dynamics}
\label{subsec:Transformed_Dynamics}

This subsection derives the unconstrained dynamics in the transformed coordinates. Based on the distance $d_i = \|p_{i-1} - p_i\|$ and its derivative $\dot{d}_i = \frac{(p_{i-1}-p_i)^T}{d_i}(v_{i-1}-v_i)$, we apply the chain rule to the spacing mapping $s_{d,i} = \ln((d_i-d_{\min})/(d_{\max}-d_i))$. The transformed spacing dynamics are obtained as:
\begin{equation}
    \dot{s}_{d,i} = \Gamma_{d,i} \frac{(p_{i-1} - p_i)^T}{d_i} (v_{i-1} - v_i)
    \label{eq:sd_dot_final}
\end{equation}
where $\Gamma_{d,i} = (d_i-d_{\min})^{-1} + (d_{\max}-d_i)^{-1}$ is the transformation gain.

Similarly, for the longitudinal velocity $v_i$, the mapped state is $s_{v,i} = \ln((v_i - v_{\min})/(v_{\max} - v_i))$. Differentiating $s_{v,i}$ and substituting the vehicle dynamics $\dot{v}_i = f_i(x_i) + g_i(x_i)u_i + d_i(t)$, the transformed velocity dynamics is derived as:
\begin{equation}
    \dot{s}_{v,i} = \Gamma_{v,i} [f_i(x_i) + g_i(x_i)u_i + d_i(t)]
    \label{eq:sv_dot}
\end{equation}
where $\Gamma_{v,i} = (v_i - v_{\min})^{-1} + (v_{\max} - v_i)^{-1}$. Consequently, the original constrained formation problem is converted into an unconstrained stabilization task.

Both $\Gamma_{d,i}$ and $\Gamma_{v,i}$ increase as the states approach the constraint boundaries. This strengthens the repulsive effect of the transformation but also increases numerical sensitivity. Therefore, the analysis is established under admissible initial conditions, and a small safety margin is used in practice to avoid excessive amplification.

\subsection{Adaptive controller design}
\label{subsec:Controller_Design}

An adaptive backstepping controller is developed to drive formation errors to zero while satisfying state constraints. For the position layer, the error derivative is $\dot{z}_{1i} = \sum_{j \in \mathcal{N}_i} a_{ij} (v_i - v_j)$. To achieve consensus, we consider the Lyapunov function $V_{1} = \frac{1}{2} \sum_{i=1}^N z_{1i}^T z_{1i}$ and design the virtual control law as:
\begin{equation}
    \alpha_i = -K_{1i} z_{1i}
    \label{eq:alpha_design}
\end{equation}
where $K_{1i}$ is a positive gain matrix. Define the transformed velocity-layer tracking error as $z_{2i} = s_{v,i} - \alpha_i$. The velocity layer dynamics then becomes: 
\begin{equation}
    \label{eq:velocity}
    \dot{z}_{2i} = \Gamma_{v,i}\big[f_i(x_i) + g_i(x_i)u_i + d_i(t)\big] - \dot{\alpha}_i .
\end{equation}

The unknown dynamics are approximated by the RBFNN such that $f_i(x_i) = W_i^T \phi_i(x_i) + \varepsilon_i$. To ensure stability, we construct the composite Lyapunov function $V = \sum_{i=1}^N (V_{1i} + \frac{1}{2} z_{2i}^T z_{2i} + \frac{1}{2} \tilde{W}_i^T \Gamma_i^{-1} \tilde{W}_i)$. The adaptive update law and the final control law are synthesized as:
\begin{align}
    \dot{\hat{W}}_i &= \Gamma_i \phi_i(x_i) z_{2i}^T - \sigma_{w,i} \hat{W}_i, \quad \sigma_{w,i} > 0 \label{eq:W_hat_dot} \\
    u_i &= g_i^{-1}(x_i) [ -K_{2i} z_{2i} - \hat{W}_i^T \phi_i(x_i) + \dot{\alpha}_i ] \label{eq:control_law}
\end{align}
where $K_{2i}$ is a positive gain matrix. Here, $\sigma_{w,i} > 0$ is a leakage coefficient that suppresses parameter drift caused by bounded disturbances and approximation errors. The modified adaptive law compensates the terms associated with $\tilde{W}_i$ and improves robustness. As a result, all closed-loop signals are UUB, and the formation errors converge to a neighborhood of the origin while satisfying the constraints.

% \begin{remark}
% Compared with the conventional gradient adaptive law, the leakage term in \eqref{eq:W_hat_dot} improves robustness and prevents excessive growth of the neural-network weights.
% \end{remark}

\subsection{Switched event-triggered mechanism}\label{subsec:SETM}

A SETM is introduced to reduce communication updates. The control input is updated only at discrete triggering instants. Let $\{t_k^i\}_{k=0}^{\infty}$ denote the discrete trigger sequence for vehicle $i$. Within the interval $t \in [t_k^i, t_{k+1}^i)$, the control input applied to the vehicle is maintained constant via a sample-and-hold process:

\begin{equation}
    u_i(t) = u_i(t_k^i)
    \label{eq:control_hold}
\end{equation}

The measurement error is defined as: 
\begin{equation}
    e_i(t) = z_{2i}(t_k^i) - z_{2i}(t)
    \label{eq:measurement_error}
\end{equation}
where $z_{2i}$ is the velocity-layer error defined in \eqref{eq:z2_constrained}. At the trigger instant $t_k^i$, it is evident that $e_i(t_k^i) = 0$.

The triggering function is defined as:  
\begin{equation}
    \Psi_i = \|e_i(t)\|^2 - \sigma_i(t) \|z_{2i}(t)\|^2
    \label{eq:trigger_function}
\end{equation}

In \eqref{eq:trigger_function}, $\sigma_i(t)$ is a dynamic switching parameter designed to adapt to different error magnitudes, defined by:
\begin{equation}
    \sigma_i(t) =
    \begin{cases}
    \delta_1, & \|z_{2i}(t)\| \ge \epsilon \\
    \delta_2, & \|z_{2i}(t)\| < \epsilon
    \end{cases}
    \label{eq:sigma_switch}
\end{equation}
where the design parameters satisfy $0 < \delta_1 < \delta_2 < 1$, and $\epsilon > 0$ represent a preset threshold. The next trigger time $t_{k+1}^i$ is determined by:
\begin{equation}
    t_{k+1}^i = \inf \left\{ t > t_k^i \mid \Psi_i(t) \ge 0 \right\}
    \label{eq:trigger_condition}
\end{equation}

By this mechanism, a smaller threshold $\delta_1$ is adopted when the tracking error is large to ensure frequent updates, whereas $\delta_2$ is used near the origin to enlarge the transmission intervals and reduce communication load without compromising stability.

Since all closed-loop signals are bounded, and $\phi_i(x_i)$, $\varepsilon_i$, $d_i(t)$, $\tilde{W}_i$, and $\Gamma_{v,i}$ are all bounded, the closed-loop dynamics in \eqref{eq:z2_dynamic_final} imply that there exists a constant $L>0$ such that $\|\dot{z}_{2i}\|\le L$. Hence, $\|e_i(t)\|\le L(t-t_k^i)$.

\subsection{Parameter selection and tuning procedure} \label{sec:param_tuning
}
For practical implementation, the selection principles of the controller gains and triggering parameters are summarized as follows.

The key parameters are selected hierarchically. First, $K_{1i}$ ensures convergence of the position-layer error. Then, $K_{2i}$ is chosen to satisfy Theorem 1, i.e., $k_{2,i} > \frac{1}{2} + \frac{\sigma}{2}$ with $\sigma = \max\{\delta_1,\delta_2\}$. The triggering parameters satisfy $0 < \delta_1 < \delta_2 < 1$: a smaller $\delta_1$ improves transient response, whereas a larger $\delta_2$ reduces communication near steady state. The threshold $\epsilon$ separates the two regimes and is chosen by balancing convergence and noise sensitivity.

\section{Stability Analysis}
\label{sec:stability}

\begin{theorem}\label{theorem1}
Consider the MASs \eqref{eq:vehicle_dynamics} under the proposed control scheme. If $k_1 > 0$ and $k_2 > \frac{1}{2} + \frac{\sigma}{2}$ with $\sigma = \max\{\delta_1, \delta_2\}$, all closed-loop signals are UUB. Moreover, the leakage term guarantees bounded adaptive weights under bounded disturbances. The tracking errors converge to a compact set around the origin, and the state constraints are strictly preserved. Furthermore, Zeno behavior is excluded.
\end{theorem}

\begin{proof}
In the backstepping design, define $z_{1i}$ as the position error. Due to the SETM in Section \ref{subsec:SETM}, the control input is updated only at triggering instants and held constant within $t \in [t_k,t_{k+1})$. This introduces the measurement error: 
\begin{equation}
    e_i(t)=z_{2i}(t_k)-z_{2i}(t).
\end{equation}

Using the RBFNN approximation $f_i(x_i)=W_i^T\phi_i+\varepsilon_i$, the closed-loop second-layer error dynamics are: 
\begin{equation}
    \dot{z}_{2i} = \Gamma_{v,i} \left[ \tilde{W}_i^T\phi_i - k_2 z_{2i} + e_i + \varepsilon_i + d_i \right]
\label{eq:z2_dynamic_final}
\end{equation}
where $\tilde{W}_i=W_i-\hat{W}_i$ denotes the weight estimation error. With the leakage term, the weight estimation error remains bounded under bounded disturbances.Consider the following global Lyapunov function candidate:
\begin{equation}
    V = \sum_{i=1}^{N} \left( \frac{1}{2} z_{1i}^T z_{1i} + \frac{1}{2} z_{2i}^T z_{2i} + \frac{1}{2} \tilde{W}_i^T \Gamma_i^{-1} \tilde{W}_i \right)
    \label{eq:V_global}
\end{equation}

Taking the time derivative of \eqref{eq:V_global} and substituting the closed-loop dynamics, we have:
\begin{equation}
\begin{split}
    \dot{V} = \sum_{i=1}^{N} \Big[ & z_{1i}^T \dot{z}_{1i} + z_{2i}^T \Gamma_{v,i} \left( \tilde{W}_i^T \phi_i - k_2 z_{2i} + e_i + \varepsilon_i + d_i \right) \\
    & - \tilde{W}_i^T \Gamma_i^{-1} \dot{\hat{W}}_i \Big]
\end{split}
\label{eq:V_dot_raw}
\end{equation}

Substituting the modified adaptive law into \eqref{eq:V_dot_raw}, the cross terms associated with the neural-network approximation are compensated. The additional leakage-related term can be upper bounded via Young's inequality and absorbed into the bounded residual term, while simultaneously preventing unbounded drift of the adaptive weights. Applying Young's inequality and the SETM triggering condition $\|e_i\|^2 \le \sigma_i \|z_{2i}\|^2$, \eqref{eq:V_dot_raw} can be upper-bounded as:
\begin{equation}
    \begin{split}
        \dot{V} \le \sum_{i=1}^{N} \Big[ &-k_1 \Gamma_{d,i} \|z_{1i}\|^2 \\
        &- \Gamma_{v,i} \left( k_2 - \frac{1}{2} - \frac{\sigma_i}{2} \right) \|z_{2i}\|^2 + C_i \Big]
    \end{split}
    \label{eq:V_dot_final}
\end{equation}
where $C_i$ is a positive constant collecting the bounds of the approximation errors, disturbances, and leakage terms. If the controller gain satisfies $k_2 > \frac{1}{2} + \frac{\sigma}{2}$ with $\sigma = \max\{\sigma_i\}$, then \eqref{eq:V_dot_final} implies:
\begin{equation}
    \dot{V} \le -\lambda V + C
    \label{eq:V_compact}
\end{equation}
where $\lambda > 0$ and $C > 0$ are design-dependent constants. Solving the inequality \eqref{eq:V_compact} yields:
\begin{equation}
    0 \le V(t) \le \frac{C}{\lambda} + \left( V(0) - \frac{C}{\lambda} \right) e^{-\lambda t}
\end{equation}

This result guarantees that all closed-loop signals, including $z_{1i}$, $z_{2i}$, and $\tilde{W}_i$, are uniformly ultimately bounded (UUB). Thus, the leakage term preserves the UUB property and guarantees bounded adaptive weights. Furthermore, since the transformed states $z_{1i}$ and $z_{2i}$ are bounded, the properties of the inverse logarithmic mapping ensure that the physical spacing and velocity remain strictly within the prescribed admissible sets.

Since all closed-loop signals are bounded, there exists a constant $L>0$ such that $\|\dot{z}_{2i}\|\le L$, and hence $\|e_i(t)\|\le L(t-t_k^i)$. If $\|z_{2i}\|\ge\epsilon$, then $\|e_i\|\ge \sqrt{\delta_1}\epsilon$ and thus $t_{k+1}^i-t_k^i\ge \sqrt{\delta_1}\epsilon/L$. If $\|z_{2i}\|<\epsilon$, similarly, $t_{k+1}^i-t_k^i\ge \sqrt{\delta_2}\epsilon/L$. Therefore,
\begin{equation}
    T_{\min}=\frac{\min\{\sqrt{\delta_1},\sqrt{\delta_2}\}\epsilon}{L}>0,
\end{equation}
which excludes Zeno behavior.
\end{proof}

	\section{Simulation}\label{sec:example}

\subsection{Simulation setup}
Simulations are conducted in MATLAB to evaluate the proposed adaptive switched ETC scheme. The formation consists of one leader and three followers, with vehicle masses $m_i$ set as $1800$ kg, $1700$ kg, $1750$ kg and $1850$ kg for $i=1,\dots,4$, respectively. The longitudinal dynamics are given by $\dot{v}_{xi} = (u_{xi}+f_{xi}+d_{xi})/m_i$, where the unknown term $f_{xi}$ is approximated by an RBFNN. The simulation lasts $50$ s with a sampling period of $0.001$ s.

 The target is a two-lane square formation with a longitudinal gap of $12$ m. The leader maintains $12$ m/s for $t < 25$ s, decelerates to $6$ m/s between $25$ s and $35$ s, and remains stable thereafter.

   \subsection{Results analysis}

\begin{figure*}
    \centering
    % 第一幅图
    \subfloat[Longitudinal tracking]{\includegraphics[width=0.33\linewidth]{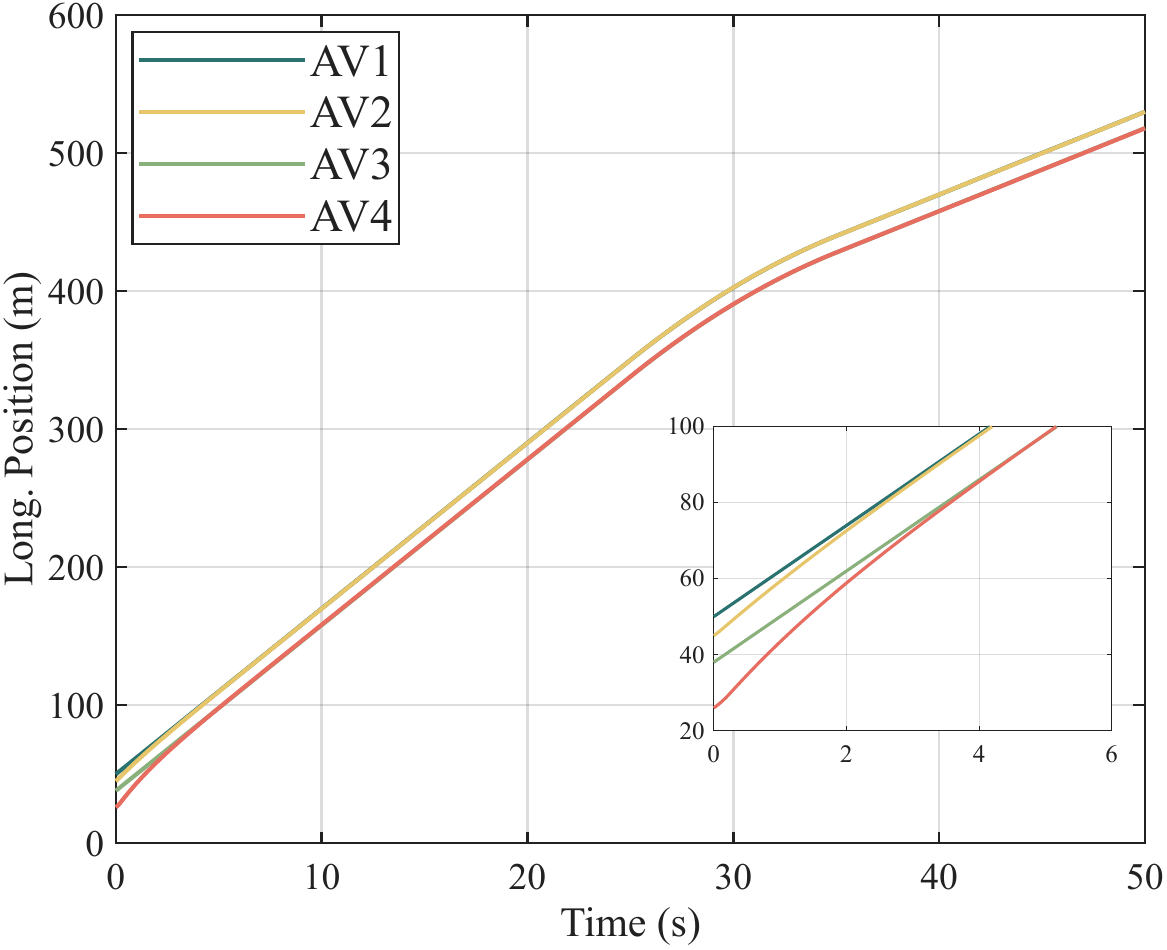}\label{fig:square_long}}
    \hfill 
    % 第二幅图
    \subfloat[Lateral tracking]{\includegraphics[width=0.33\linewidth]{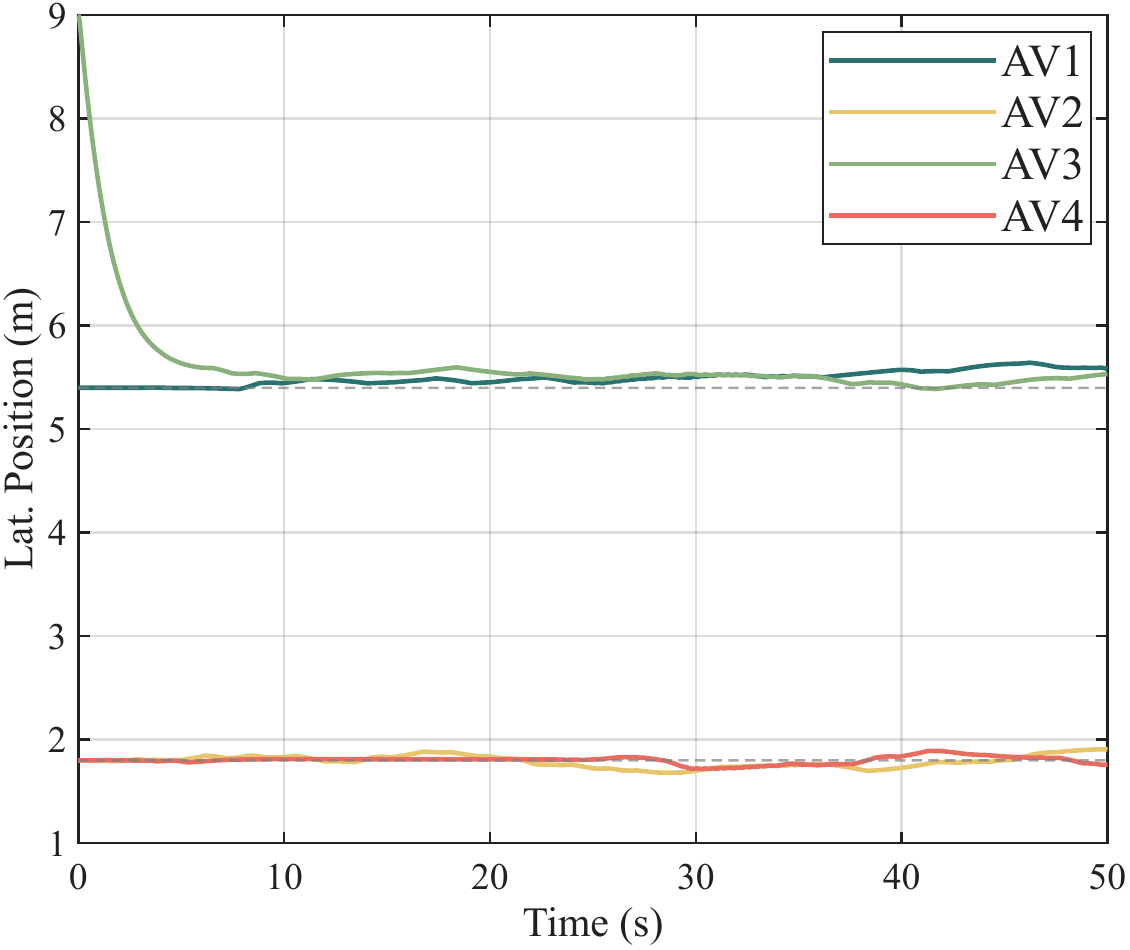}\label{fig:square_lat}}
    \hfill
    % 第三幅图：
    \subfloat[Safe distances]{\includegraphics[width=0.33\linewidth]{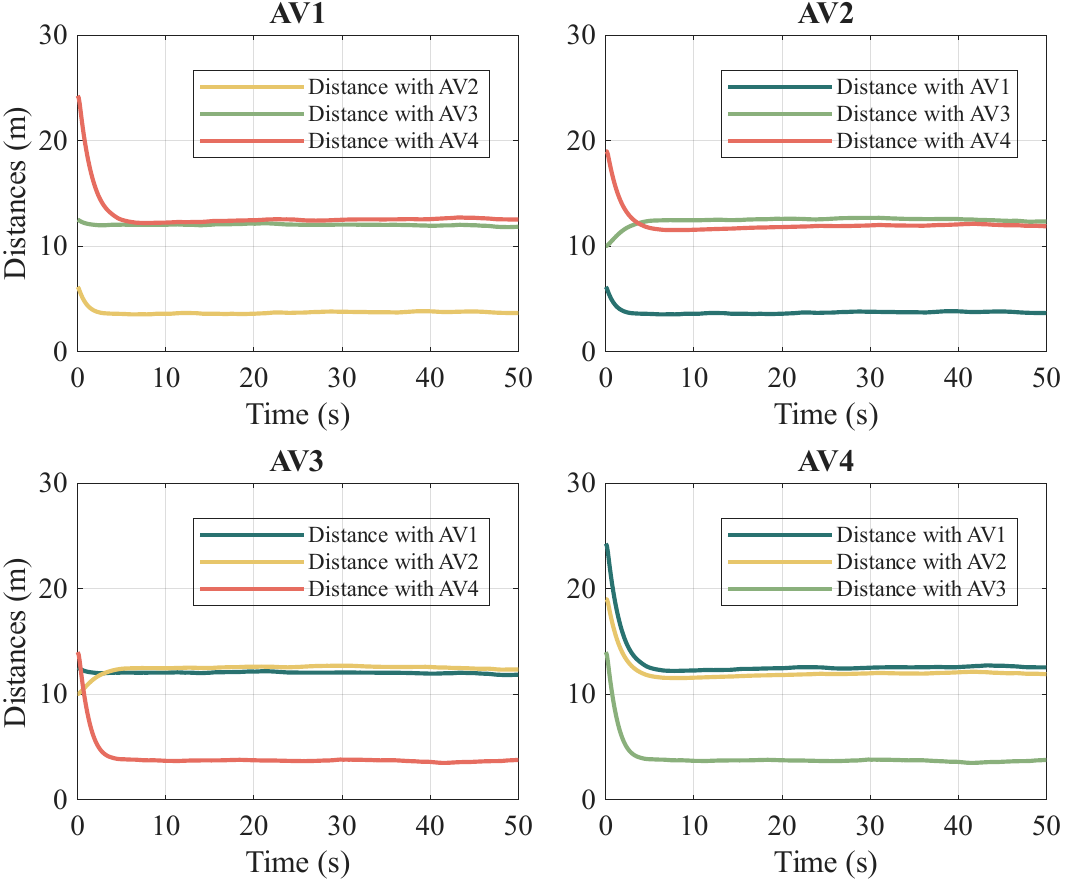}\label{fig:square_dist}}
    
    \caption{Longitudinal and lateral tracking performance of the vehicle formation control and safe distances in control process.}
    \label{tracking}
\end{figure*}

Figures \ref{fig:square_long} and \ref{fig:square_lat} illustrate the longitudinal and lateral tracking performance.
Despite initial positional deviations, the vehicles achieved the preset following distance of 12 metres within approximately 4 seconds.
When the lead vehicle decelerated, the convoy remained stable, indicating accurate tracking performance.

Figure \ref{fig:square_dist} illustrates the trend in inter-vehicle distances.
Throughout the simulation, all inter-vehicle distances remained above the safety threshold, confirming that the collision avoidance function and the formation topology were maintained.

\begin{figure*}[htbp]
\centering

\subfloat[Longitudinal velocity]{
    \includegraphics[width=0.48\linewidth]{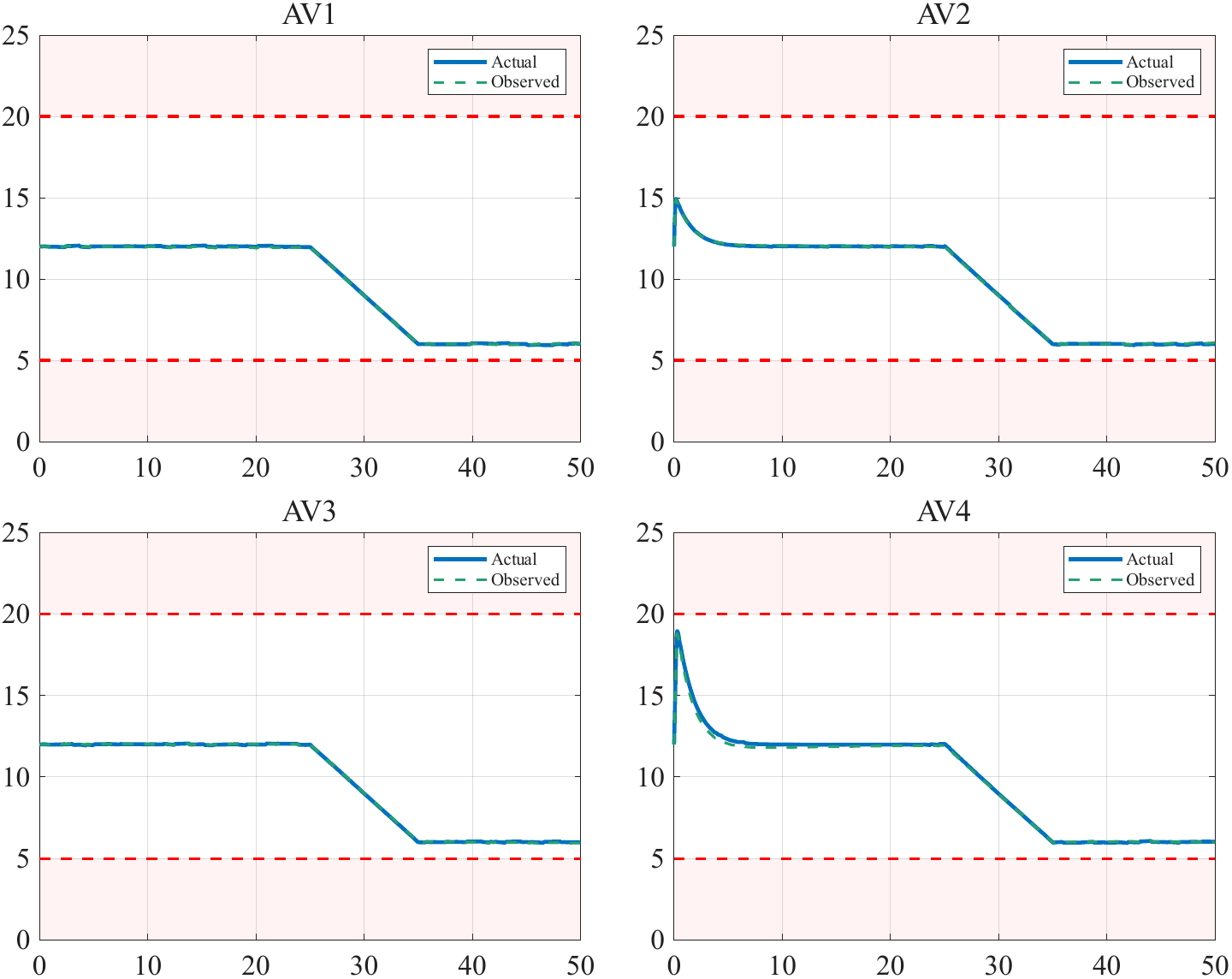}
    \label{fig:long_audit}
}
\hfill
\subfloat[Lateral velocity]{
    \includegraphics[width=0.48\linewidth]{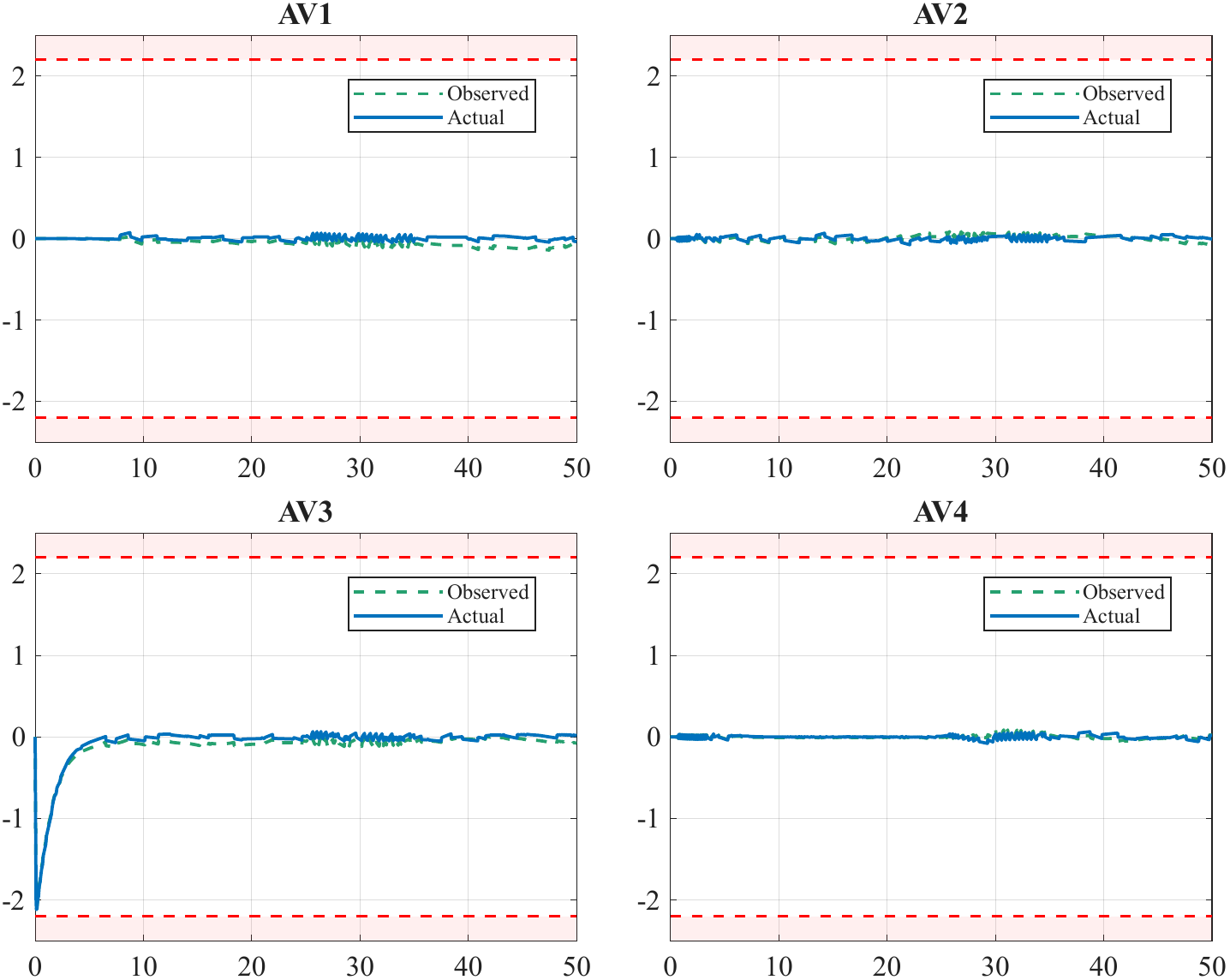}
    \label{fig:lat_audit}
}

\caption{Safety admissibility audit for longitudinal and lateral velocities of the AV formation.}
\label{fig:speed_audit}

\end{figure*}

Figures \ref{fig:long_audit} and \ref{fig:lat_audit} show the longitudinal and lateral velocity responses, respectively. Throughout the simulation, both the actual and observed velocities remained within the predefined constraints. Even during the deceleration phase of the lead vehicle, the velocity trajectories remained within the permitted range, confirming that the proposed diffeomorphic mapping is effective in enforcing the velocity constraints.

\begin{figure}
    \centering
    \includegraphics[width=0.99\linewidth]{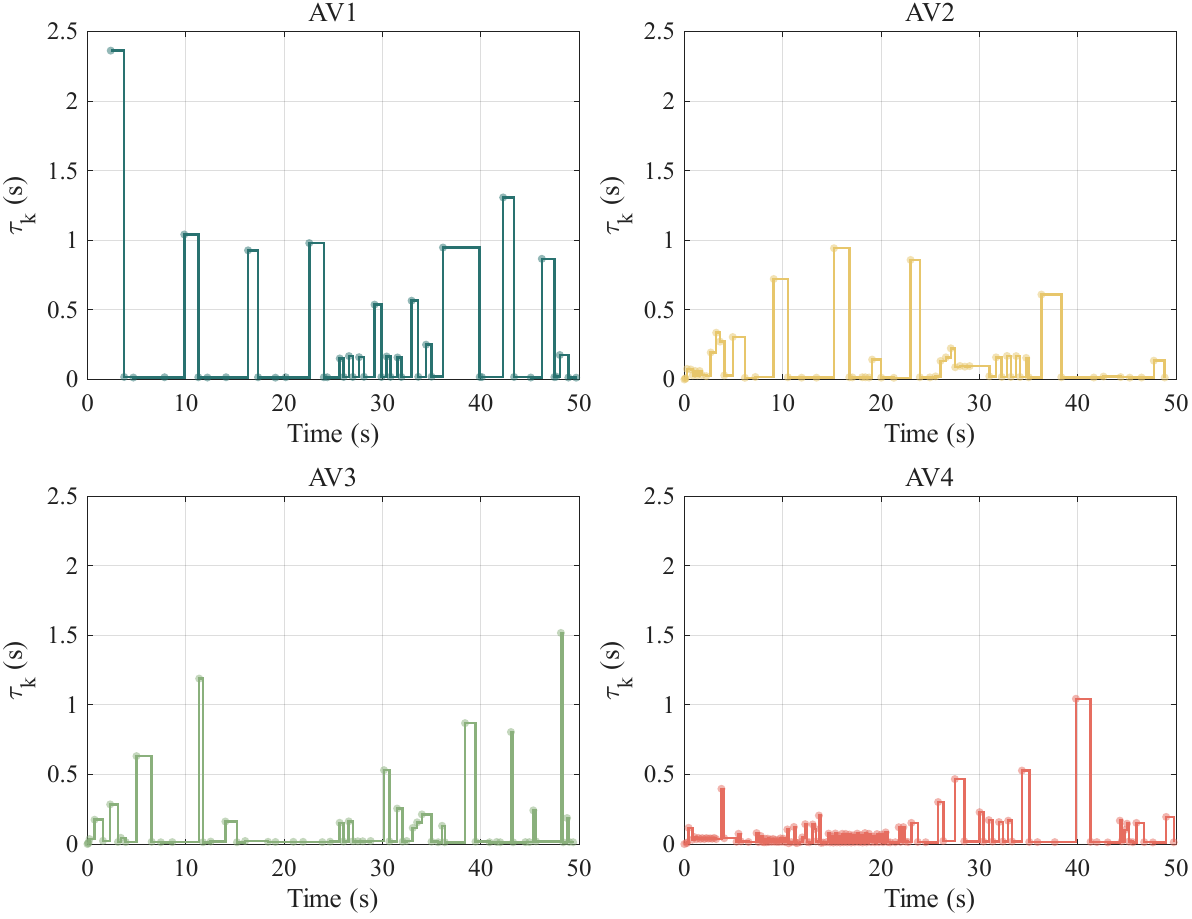}
    \caption{The Longitudinal update time interval of SETM law
for the AVs in square formation control process.}
    \label{fig:interval}
\end{figure}

\begin{table}
    \centering
\caption{Control signal update statistics.}
    \begin{tabular}{ccc}
    \hline
    Vehicle & Total Triggers& Reduction (\%)\\
    \hline
    AV1 & 212& 99.58\\
    AV2 & 301& 99.40\\
    AV3 & 306& 99.39\\
    AV4 & 678& 98.64\\
    \hline
    \end{tabular}
    \label{tab:triggering_ratio}
\end{table}

Figure \ref{fig:interval} shows the trigger intervals generated by SETM. During the transient phase, these intervals become shorter; during steady-state operation, they become longer.

Table \ref{tab:triggering_ratio} lists the total number of trigger events for each vehicle. Compared to periodic control, the communication load is reduced by more than 98\%. Furthermore, all trigger intervals are strictly maintained within the positive range, which is consistent with the results of the theoretical analysis of the no-Zeno phenomenon.

\begin{remark}
Compare with \cite{ref_wang_arxiv_2025}, the proposed framework explicitly incorporates state-constraint conditions by imposing limits on vehicle speed, thereby replicating realistic scenarios where speed restrictions apply to moving vehicles. This imposes higher demands on the control process.
\end{remark}

\section{Conclusion}\label{sec:conclusion}

This paper studies vehicular formation control with full-state constraints. A smooth nonlinear mapping is introduced to transform constrained states into unconstrained coordinates, thereby avoiding control singularities near the constraint boundaries. Based on this, an adaptive backstepping controller with an RBFNN is developed to handle unknown nonlinear dynamics, and a SETM is designed to reduce communication updates. Theoretical analysis shows that all closed-loop signals are UUB, the adaptive weights remain bounded under bounded disturbances, and Zeno behavior is excluded. Simulations verify the effectiveness of the proposed method in formation tracking, constraint satisfaction, and communication reduction.

\end{document}